\documentclass[11pt,a4paper]{article}
\usepackage[margin=2.5cm]{geometry}
\usepackage[T1]{fontenc}
\usepackage{graphicx} 
\usepackage{xcolor}
\usepackage[ruled,algo2e]{algorithm2e}
\usepackage{algorithm}
\usepackage{algorithmic}
\SetKwInput{KwOutput}{Output}
\usepackage{amssymb} 
\usepackage{amsthm}
\usepackage{amsmath} 

\usepackage{bookmark}
\usepackage{thm-restate}

\newtheorem{theorem}{Theorem}
\newtheorem{lemma}{Lemma}

\date{}

\title{\Large Semi-Robust Communication Complexity of Maximum Matching}
    \author{Gabriel Cipriani Huete\footnotemark[2]
    \and Adithya Diddapur\footnotemark[2]
    \and Pavel Dvo\v{r}\'{a}k\thanks{\texttt{koblich@iuuk.mff.cuni.cz}, Charles University, Prague, Czech Republic. Supported by Czech Science Foundation GAČR grant \#22-14872O.}    
    \and Christian Konrad\thanks{\texttt{\{gabriel.cipriani,adi.diddapur,christian.konrad\}@bristol.ac.uk}, School of Computer Science, University of Bristol, Bristol, UK.}}

\DeclareMathOperator{\poly}{poly}
\DeclareMathOperator{\Exp}{\mathbb{E}}

\begin{document}

\maketitle
\begin{abstract}
We study the one-way two-party communication complexity of \textsf{Maximum Matching} in the {\em semi-robust setting} where the edges of a maximum matching are randomly partitioned between Alice and Bob, but all remaining edges of the input graph are adversarially partitioned between the two parties. 

\vspace{0.1cm}

We show that the simple protocol where Alice solely communicates a lexicographically-first maximum matching of their edges to Bob is surprisingly powerful: We prove that it yields a $3/4$-approximation in expectation and that our analysis is tight.

\vspace{0.1cm}

The semi-robust setting is at least as hard as the {\em fully robust} setting. In this setting, all edges of the input graph are randomly partitioned between Alice and Bob, and the state-of-the-art result is a fairly involved $5/6$-approximation protocol that is based on the computation of edge-degree constrained subgraphs [Azarmehr, Behnezhad, ICALP'23]. Our protocol also immediately yields a $3/4$-approximation in the fully robust setting. One may wonder whether an improved analysis of our protocol in the fully robust setting is possible: While we cannot rule this out, we give an instance where our protocol only achieves a $0.832 < 5/6 = 0.8\overline{3}$ approximation. Hence, while our simple protocol performs surprisingly well, it cannot be used to improve over the state-of-the-art in the fully robust setting.

\end{abstract}

\section{Introduction}
\paragraph{One-way Two-Party Communication.} In the {\em one-way two-party communication} setting \cite{y79}, two parties, denoted Alice and Bob, each hold a portion of the input. Based on their input, Alice sends a message to Bob, who, upon receipt, computes the output of the protocol as a function of his input and the message received. Protocols can also be randomized, in which case Alice and Bob additionally have access to both private and shared infinite sequences of random bits. 
The goal is to design protocols with small {\em communication cost}, i.e., protocols that communicate as few bits as possible.

The one-way two-party communication setting is a clean and interesting model that deserves attention in its own right.
However, due to its connection to streaming algorithms, it has received an additional significant boost in interest in the last decade or so. We will elaborate more on the connection between one-way communication and streaming in Subsection~\ref{sec:connection-streaming}. %

In this work, we study the \textsf{Maximum Matching} problem in the one-way two-party communication setting. Given a graph $G=(V, E)$ with $|V| = n$, a {\em matching} $M \subseteq E$ is a subset of vertex-disjoint edges, and a {\em maximum matching} $OPT$ is one of largest cardinality. In the one-way two-party setting, Alice and Bob each hold subsets of the edges $E_A, E_B \subseteq E$ of the input graph, respectively, and the objective for Bob is to output an $\alpha$-approximation to \textsf{Maximum Matching}, i.e., a matching $M$ such that $|M| \ge \alpha \cdot |OPT|$, for some $0 < \alpha \le 1$, where $OPT$ denotes a maximum matching in the input graph $G=(V, E_A \cup E_B)$.
  
It is known that computing an exact maximum matching requires Alice to send a message of size $\Omega(n^2)$ \cite{fkmsz05}. Goel et al. \cite{gkk12} were the first to consider approximations in this setting and gave a communication protocol with approximation factor $2/3$ that communicates $O(n \log n)$ bits. They also proved that a message of size $n^{1 + \Omega\left(\frac{1}{\log \log n}\right)}$ is necessary for going beyond such an approximation factor. 

\paragraph{The Robust Setting.} Assadi and Behnezhad \cite{ab21a} initiated the study of \textsf{Maximum Matching} in the {\em robust} one-way two-party communication setting \cite{ccm16}. In the robust setting, each edge of the input graph is randomly assigned to either Alice or Bob, each with probability $1/2$. They proved in a non-constructive way that there exists a protocol for bipartite graphs that achieves a $0.716$-approximation in expectation (over the random edge partitioning). Azarmehr and Behnezhad \cite{ab23} subsequently gave a protocol for general graphs that achieves a $(5/6-\epsilon)$-approximation, for any $\epsilon > 0$, which constitutes the state-of-the-art bound. Regarding lower bounds, it is known that computing a $(1-\epsilon)$-approximation requires a message of size $\Omega(n \poly \log(n) \cdot \exp(\frac{1}{\epsilon^{0.99}}))$ \cite{ab21}, but, for example, it has not yet been ruled out that a $0.999$-approximation with a message of size $O(n \poly \log n)$ exists. 

Azarmehr and Behnezhad's protocol is an implementation of the $(2/3-\epsilon)$-approximation random order one-pass streaming algorithm by Bernstein \cite{b20} in the robust one-way two-party communication setting. The communication setting then allows for an improved analysis in order to boost the approximation factor  from $2/3-\epsilon$ to $5/6 - \epsilon$. Bernstein's algorithm is based on the computation of edge-degree constrained subgraphs (EDCS) \cite{bs15}, which are known to do well in various resource-constraint models (see, for example,  \cite{ab15,abbms19,ab19,ab21,bk22,b24}).  

The EDCS technique is arguably fairly involved, and our aim is to understand whether such heavy machinery is necessary to obtain well-performing protocols in the robust setting. In fact, it is not even clear how well the perhaps most natural protocol for this problem where Alice simply sends a largest possible matching among their input edges to Bob performs. %
It is easy to see that this protocol immediately yields a $1/2$-approximation, however, it is unclear  whether it performs any better.

\paragraph{The Semi-Robust Setting.} In this work, we study this simple protocol, i.e., Alice sends a (lexicographically-first) maximum matching of their edges to Bob, in a slightly more challenging variant of the robust setting that we refer to as the {\em semi-robust setting}. Let $G=(V, E)$ be the input graph and let $OPT$ be an arbitrary maximum matching in $G$. Then, in the semi-robust setting, the edges of $E \setminus OPT$ are arbitrarily, potentially adversarially, partitioned between Alice and Bob, and the edges of $OPT$ are partitioned uniformly and independently at random between Alice and Bob. The semi-robust setting is at least as hard as the fully robust setting since any protocol designed for the semi-robust setting also constitutes a protocol with the same guarantees in the fully robust setting.

A further motivation for studying this setting is with regards to the $2/3$-approximation barrier for semi-streaming protocols under fully adversarial partitions\cite{gkk12}.
The previously mentioned work in the fully robust setting shows that this barrier can be broken under a fully random partition of the input.
The semi-robust setting then allows us to ask a narrower question - how brittle is the $2/3$-approximation barrier against input partitions which are only partly adversarial?

\subsection{Our Results}
In the following, we denote by $\Pi_{\text{LFMM}}$ the protocol where Alice sends a lexicographically-first maximum matching of their edges to Bob, see Algorithm~\ref{alg:pilfmm}.
Note that we allow the (lexicographic) ordering on vertices to be adversarially chosen, and the only randomness considered later is over the set of possible partitions.

\begin{algorithm}
 \begin{algorithmic}
  \REQUIRE Input graph $G=(V, E)$, Alice holds $E_A \subseteq E$ and Bob holds $E_B \subseteq E$ \vspace{0.2cm}
  \STATE 
  \hspace{-0.4cm} \begin{tabular}{ll}
\textbf{Alice:} &   Compute lexicographically-first maximum matching $M$ of $(V, E_A)$, send $M$ to Bob \\
\textbf{Bob:}  & Output a maximum matching in $(V, E_B \cup M)$
  \end{tabular}
 \end{algorithmic}
 \caption{Protocol $\Pi_{\text{LFMM}}$ \label{alg:pilfmm}}
\end{algorithm}
As our main result, we show that $\Pi_{\text{LFMM}}$ is surprisingly strong and yields an approximation factor of $3/4$ in expectation in the semi-robust setting.

\begin{theorem}
 In the semi-robust setting, the protocol $\Pi_{\text{LFMM}}$ constitutes a one-way two-party communication protocol that achieves an approximation factor of $3/4$ in expectation.
\end{theorem}

We note that our analysis is tight, in that there are simple graphs $G=(V, E)$, e.g., the path on three edges, together with an adversarial partitioning of $E \setminus OPT$ between Alice and Bob so that $\Pi_{\text{LFMM}}$ does not yield an expected approximation factor above $3/4$. 

We note that our analysis can easily be adapted to the protocol where Alice sends a lexicographically-first {\em maximal matching} to Bob. This protocol achieves an approximation factor of $5/8$ in expectation, which is also tight.

One may wonder whether a tailored analysis of the $\Pi_{\text{LFMM}}$ protocol to the fully robust setting allows for an improved approximation factor. While we cannot rule this out, via an exhaustive search using a computer, we identified an example graph together with an adversarial vertex ordering such that $\Pi_{\text{LFMM}}$ yields an expected approximation factor of $0.832 < 5/6 = 0.8\overline{3}$ in the fully robust setting (see Figure~\ref{fig:hard-example} for details). This example demonstrates that the $\Pi_{\text{LFMM}}$ protocol cannot achieve the guarantees established by Azarmehr and Behnezhad \cite{ab23} who used the more involved EDCS technique.

\begin{figure}
\begin{center} 
\includegraphics[height=2.4cm]{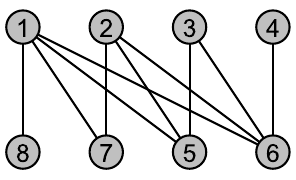} 
\end{center}
 \caption{The hard input graph for $\Pi_{\text{LFMM}}$ in the fully robust setting is the semi-complete graph on $8$ vertices. Importantly, a specific vertex ordering is required that defines the lexicographical ordering of the maximum matchings computed by Alice (see the preliminaries for details). The ordering is indicated by the integers assigned to the vertices. We observe that the positions of $5$ and $6$ are required to be as they are -- the more intuitive ordering with $5$ and $6$ exchanged does not yield an expected approximation factor $< 5/6$. \label{fig:hard-example}}  
\end{figure}

\subsection{Our Techniques}
We will now discuss the main ideas underlying our analysis of the $\Pi_{\text{LFMM}}$ protocol in the semi-robust setting. To this end, we denote by $M$ the lexicographically-first maximum matching of the edges $E_A$ sent from Alice to Bob, and by $OPT_A$ and $OPT_B$ Alice's and Bob's parts of the randomly partitioned edges of $OPT$, respectively. 

In our analysis, we consider the graphs $H = (V, M \cup OPT)$ and $H' = (V, M \cup OPT_B)$ and prove that a maximum matching in $H'$ constitutes a $3/4$-approximation of $OPT$ in expectation. %

We observe that the graph $H$ constitutes a collection of disjoint paths and cycles since it is the union of two matchings. The structures responsible for small approximation factors are short paths of odd length that consist of one more $OPT$ edge than $M$ edges -- these paths are referred to as {\em augmenting paths}. We will now argue that Bob holds all the $OPT$ edges on an augmenting path more often than one would a priori expect, which allows Bob to improve over the matching edges $M$. This argument alone is sufficient to obtain a $2/3$-approximation in expectation. We show, however, as well that, overall, the expected number of short augmenting paths in $H$ can be non-trivially bounded, which enables us  to further boost the approximation factor to $3/4$. 

\paragraph{Skewed Distribution of $OPT$ Edges}

For simplicity, consider a length-$3$ augmenting path $P = e_1, e_2, e_3$ in $H$ such that $e_1, e_3 \in OPT$ and $e_2 \in M$. The argument extends naturally to longer paths. We would like to argue as follows: Since the distribution of $OPT$ edges between Alice and Bob is uniformly random, with probability $\frac{1}{4}$, both $OPT$ edges $e_1, e_3$ are contained in Bob's set $OPT_B$. Hence, Bob can augment this path and no loss is incurred. If this statement was true then, in expectation, Bob could output a path of length $\frac{1}{4} \cdot 2 + \frac{3}{4} \cdot 1 = \frac{5}{4}$, or, put differently, the expected approximation factor on a path of length $3$ would be $\frac{5}{4 \cdot 2} = \frac{5}{8}$ since the optimal solution on this path is $2$. This argument is of course flawed since, conditioned on the path $P$ existing, the distribution of $e_1, e_3 \in OPT$ may no longer be uniform, and we prove that this is indeed the case. However, conditioned on $P$ existing, the distribution of $OPT$ edges is in fact {\em skewed to our advantage}: Observe that, if $P$ exists, it is not possible that both $e_1$ and $e_3$ are contained in $OPT_A \subseteq E_A$ since this would contradict the fact that $M$ is a maximum matching in $E_A$. Hence, instead of the four possibilities as to how the edges $e_1$ and $e_3$ are partitioned between Alice and Bob, at most three possibilities remain. Observe that if exactly these three possibilities remained and their distribution was uniform then we have made a significant gain since then with probability $1/3$, both $OPT$ edges are located at Bob's end, which allows Bob to output a matching of expected size $\frac{1}{3} \cdot 2 + \frac{2}{3} \cdot 1 = \frac{4}{3}$ on path $P$, i.e., we obtain an expected approximation factor of $2/3$ on $P$. 

The key part of our analysis of this claim is to show that the uniform distribution over the remaining three cases indeed constitutes the worst case. To this end, we observe that if a partitioning of the $OPT$ edges into $OPT_A$ and $OPT_B$ leads to the existence of an odd-length path $P$ of length $2i-1$, for some $i \ge 2$, then, if we moved any of the $OPT$ edges of $P \cap OPT_A$ to $OPT_B$, then the same path $P$ would exist. In other word, conditioned on $P$ existing, the family of possible sets $OPT_B \cap P$ is {\em closed under taking super sets}. We exploit this key property to establish the claimed bound. 

To summarize this part of the argument, the skewed nature of the distribution of $OPT$ edges conditioned on a specific short augmenting path existing is sufficient to argue a $2/3$-approximation in expectation. To improve this bound further, we argue that the random partitioning of $OPT$ edges implies that the number of short augmenting paths is non-trivially bounded.

\paragraph{Bounded Number of Short Augmenting Paths}
So far we have argued that if $H$ contains an augmenting-path of length $3$ then with probability at least $1/3$, Bob holds both optimal edges contained in this path. We also already argued that Bob must hold at least one optimal edge on this path since otherwise the matching $M$ would not be maximum. Phrased differently, the expected number of optimal edges held by Bob on a path of length $3$ is at least $\frac{1}{3} \cdot 2 + \frac{2}{3} \cdot 1 = \frac{4}{3}$.

This insight allows us to bound the expected total number of length-$3$ paths in $H$ that we denote by $n_3$: Since, overall, we expect Bob to hold half of the optimal edges, i.e., $\Exp |OPT_B| = \frac{1}{2} |OPT|$, and we also argued $\Exp |OPT_B| \ge \frac{4}{3} \cdot n_3$, we obtain that $n_3 \le \frac{3}{8} |OPT|$. This is a significant gain over the worst-case bound $n_3 \le \frac{1}{2} |OPT|$, which happens when all edges of $M$ are contained in $3$-augmenting paths. 

Using similar arguments, we establish a condition that jointly bounds the expected number of length $5,7, 9, \dots$ augmenting paths in $H$. We then give a linear program that captures this condition as well as further constraints so that  the objective value of the LP lower bounds the worst-case expected approximation factor of the protocol $\Pi_{\text{LFMM}}$. We identify that both the graphs $H_3$ and $H_5$ consisting of a single length-$3$ and a single length-$5$ augmenting path, respectively, constitute worst-case examples with expected approximation factor $3/4$.

\subsection{Connection to Streaming} \label{sec:connection-streaming}
In the {\em one-pass streaming model} for graph problems, the edges of an input graph are presented to an algorithm in arbitrary order one-by-one, and the objective is to design algorithms that use as little space as possible. It is easy to see that every one-pass streaming algorithm with space $s$ can be used as a one-way two-party communication protocol with communication cost $s$. The converse, however, is not true. Hence, the one-pass streaming model is a harder model than the one-way two-party communication setting. 

The one-pass {\em random order} streaming setting, where the edges arrive in uniform random order, is closely related to the robust one-way two-party communication setting. Konrad et al. \cite{kmm12} were the first to study \textsf{Maximum Matching} in this setting and showed that improvements over the approximation factor $\frac{1}{2}$ are possible using space $O(n \log n)$. Further improvements were given by \cite{k18}, \cite{gkms19}, and \cite{fhmrr20} until Bernstein's breakthrough result \cite{b20}, who established an approximation factor of $2/3$ using the EDCS-technique. Assadi and Behnezhad then combined the EDCS-technique with augmenting paths, giving a $(2/3+\epsilon)$-approximation, for a fixed small constant $\epsilon > 0$, which constitutes the state-of-the-art result in this setting. 

\subsection{Outline}
In Section~\ref{sec:prelim}, we provide details as to how a lexicographically-first maximum matching is obtained. Then, in Section~\ref{sec:protocol}, we prove our main result, i.e., that the $\Pi_{\text{LFMM}}$ protocol achieves an approximation factor of $3/4$ in expectation, and we also argue that sending a lexicographically-first maximal matching yields an approximation factor of $5/8$. Finally, we conclude in Section~\ref{sec:conclusion} with directions for further research.

\section{Preliminaries} \label{sec:prelim}
\noindent \paragraph{Lexicographically-first Maximum Matching}
Let $G=(V, E)$ be any graph on $n$ vertices. We assume that there exists an arbitrary but fixed ordering of $V$ (e.g., integral vertex identifiers), i.e., we have access to an ordering $\sigma: V \rightarrow \{1, 2, \dots, n\}$ of the vertices of $G$. 
The ordering $\sigma$ then allows us to define an ordering of the edges $E$ so that, for two distinct edges $e_1=(u_1, v_1)$ with $\sigma(u_1) < \sigma(v_1)$ and $e_2 = (u_2, v_2)$ with $\sigma(u_2) < \sigma(v_2)$ we have $e_1 < e_2$ if $\sigma(u_1) < \sigma(u_2)$ or $(\sigma(u_1) = \sigma(u_2)$ and $\sigma(v_1) < \sigma(v_2))$.
This ordering is then extended to maximum matchings. For two distinct maximum matchings $M_1 = \{e_1, e_2, \dots \}$ with $e_1 < e_2 < \dots$ and $M_2 = \{f_1, f_2, \dots \}$ with $f_1 < f_2 < \dots$ we have $M_1 < M_2$ if there exists an index $j \in \{1, 2, \dots \}$ such that $e_i = f_i$ holds, for all $i < j$, and $e_j < f_j$. Then, we say that $M$ is a {\em lexicographically-first maximum matching} in $G$ if, for every other maximum matching $M'$, $M < M'$ holds.

\section{Our Protocol} \label{sec:protocol}
In this section, we analyse the $\Pi_{\text{LFMM}}$ protocol where Alice sends the lexicographically-first maximum matching to Bob. To this end, we first give further notation required for our analysis in Subsection~\ref{sec:notation}, prove our main technical lemma in Subsection~\ref{sec:main-tech}, give non-trivial bounds on the output size and on the size of an optimal solution in Subsection~\ref{sec:bounds-out-opt}, and finally prove our main theorem in Subsection~\ref{sec:main-thm}.

\subsection{Further Notation}\label{sec:notation}
Let $G=(V, E)$ denote the input graph, and let $OPT$ be an arbitrary but fixed maximum matching. We denote by $OPT_A$ and $OPT_B$ the subsets of $OPT$ that are assigned to Alice and Bob, respectively. We also write $E_A$ to denote the set of edges held by Alice (including $OPT_A$), and by $E_B$ the set of edges held by Bob (including $OPT_B$). We also write $\pi = (OPT_A, OPT_B)$ to denote the partitioning of $OPT$ edges into $OPT_A$ and $OPT_B$. We denote by $M$ the matching sent from Alice to Bob.

We consider the graph $H = (V, M \cup OPT)$, i.e., the graph spanned by the matching sent from Alice to Bob and the optimal edges, and its subgraph $H' = (V, M \cup OPT_B)$. Note that Bob knows the entire subgraph $H'$. We will prove that $H'$ contains a large matching in expectation. We denote by $OUT$ the matching produced by Bob.

Since the edge set of graph $H$ is the union of two matchings, the set of connected components $\mathcal{C} = \{\mathcal{C}_1, \mathcal{C}_2, \dots \}$ of $H$ consists of individual edges, paths, and cycles. We will see that, on all components of even size and on components consisting of individual edges, Bob can output optimal matchings. The key part of the analysis is to show that, even on odd-length paths,  the probability that Bob holds all optimal edges and can thus output a locally optimal matching is non-trivially bounded from below.

Given a set of components $C \in supp(\mathcal{C})$, for $i \ge 2$, let $n_i(C)$ denote the number of length $2i - 1$ augmenting paths in $C$. We also denote by $n_1(C)$ the number of optimal edges contained in even-length paths, in cycles, or that appear as isolated edges in $H$. The quantity $n_1(C)$ captures the components on which Bob's output matching is trivially optimal. 

In the following, we consider expectations and probabilities over $\pi \mid (\mathcal{C}=C)$ that we denote by $\pi \mid C$ in short, i.e., the marginal distribution of $\pi$ when the components $C$ are established within $H$. Using this notation, we first define the quantities $n_i$, for every $i \ge 1$, as follows:

\begin{align}
 n_i & = \Exp_{\pi} n_i(\mathcal{C}) = \sum_{C \in supp(\mathcal{C})} \Pr[\mathcal{C} = C] \cdot  n_i(C) \ . \label{eqn:966}
\end{align}

\subsection{Main Technical Lemmas} \label{sec:main-tech}
We are now ready to prove our main technical lemmas, which show that, for each odd-length path in $\mathcal{C}$, there is a non-trivial probability that Bob holds all optimal edges on that path, and that, for each odd-length path, we expect Bob to hold more than half of the optimal edges along that path.

\begin{lemma}\label{lem:lex-first}
 Let $G=(V, E)$ be a graph, and let $M$ denote the lexicographically-first maximum matching in $G$. Then, for every subset $F \subseteq E \setminus M$, the matching $M$ is also the lexicographically-first maximum matching in $(V, E \setminus F)$.
\end{lemma}
\begin{proof}
 For the sake of a contradiction, suppose that $M' \neq M$ was the lexicographically-first maximum matching in $(V, E \setminus F)$. Then, $M'$ is lexicographically {\em smaller} than $M$, which is a contradiction to the fact that $M$ is the lexicographically-first maximum matching in $(V, E)$. 
\end{proof}

We now use Lemma \ref{lem:lex-first} to prove Lemma \ref{lem:main}.

\begin{lemma} \label{lem:main}
 Let $C \in supp(\mathcal{C})$ be any element, and let $P \in C$ be any odd length path of length $2k-1$.  Then:
 \begin{enumerate}
  \item $\Exp_{\pi|C} |OPT_B \cap P| \ge k \frac{2^{k-1}}{2^k - 1}$, and 
  \item $\Pr_{\pi|C} \left[|OPT_B \cap P| = k \right] \ge \frac{1}{2^k - 1}$.
 \end{enumerate}
\end{lemma}

\begin{proof}
Let $OPT_P$ denote the optimal edges that are contained in $P$.

We partition the support of $\pi|C$ into the $k$ sets $\Pi_1, \dots, \Pi_k$ such that $\pi \in supp(\pi|C)$ is put into $\Pi_j$ if and only if Bob receives exactly $j$ edges of $OPT_P$ under $\pi$.

We first argue that $\Pi_1, \dots, \Pi_k$ is indeed a partitioning of $supp(\pi|C)$, or, equivalently, for any $\pi \in supp(\pi|C)$, Bob holds at least one edge. To see this, suppose that Bob did not hold any such edge under $\pi$. Then, Alice holds all optimal edges within component $P$, i.e., $OPT_P \subseteq OPT_A$. This, however, is a contradiction to the fact that $M$ is a maximum matching within $E_A$ since $M$ could have been augmented by Alice using the edges $OPT_P \subseteq E_A$. Hence, Bob holds at least one edge of $OPT_P$, and $\Pi_1, \dots, \Pi_k$ is a partition of $supp(\pi|C)$.

Next, consider any $\pi \in \Pi_j$, for any $j < k$. We observe that every $\pi'$ that is obtained from $\pi$ by moving one of the $k-j$ $OPT_P \cap OPT_A$ edges held by Alice to Bob is contained in $\Pi_{j+1}$. This immediately follows from Lemma~\ref{lem:lex-first}. 

Next, we consider the directed acyclic graph $G_C = (supp(\pi|C), E_C)$ consisting of $k$ layers, where the layers constitute the sets $\Pi_1, \dots, \Pi_k$. Edges only exist between consecutive layers $\Pi_i$ and $\Pi_{i+1}$, for any $1 \le i \le k-1$, and are directed from layer $i$ to $i+1$. Let $\pi_1 \in \Pi_i$ and $\pi_2 \in \Pi_{i+1}$. Then, there exists an edge between $\pi_1$ and $\pi_2$ if and only if $\pi_2$ can be obtained from $\pi_1$ by moving one of Alice's $OPT_P$ edges to Bob. 

We will now argue the following inequality: 
\begin{align*}
|\Pi_{i+1}| \ge |\Pi_{i}| \cdot \frac{k - i}{i+1} \ .
\end{align*}
To this end, we investigate the out-degrees and in-degrees of the vertices in $G_{C}$. Let $\pi_1 \in \Pi_i$ be any vertex. Then, by the argument above, $\pi_1$ has an out-degree of $k-i$ since we established that moving any $OPT_P$ edge held by Alice to Bob yields a partitioning that is also in the support of $\pi|C$ and thus also in $\Pi_{i+1}$. Next, consider any $\pi_2 \in \Pi_{i+1}$. Then, $\pi_2$ has an in-degree of at most $i+1$ since any $\pi \in \Pi_i$ that has an edge towards $\pi_2$ can be obtained by moving one $OPT_P$ edge held by Bob to Alice. Since Bob holds only $i+1$ $OPT_P$ edges there can only be $i+1$ different such vertices. This implies the claimed inequality.

Next, let 
$$p_i := \Pr_{\pi \mid C} [| OPT_B \cap OPT_P| = i]$$
denote the probability that Bob holds exactly $i$ edges of $OPT_P$. The previous arguments give us the following constraints: 
\begin{align}
p_{i+1} & \ge p_i \cdot \frac{k-i}{i+1} \ , & \text{constraint } \mathbf{C1}  \label{eqn:01} \\
\nonumber p_0 & = 0 \ , \text{and} \\
 \sum_{i \ge 1}^k p_i & = 1 \ . & \text{constraint } \mathbf{C2}  \label{eqn:03}
\end{align}

For any $i \in [k-1]$, we apply $\mathbf{C1}$ $k-i$ times, and we obtain the bound:
\begin{align}
 p_i & \le p_{i+1} \frac{i+1}{k-i} \le p_{i+2} \frac{i+2}{k-i-1} \cdot \frac{i+1}{k-i} \nonumber 
 \le \dots \le \\
 & \le p_k \cdot \frac{k \cdot (k-1) \cdot \ldots \cdot (i+1)}{1 \cdot \ldots \cdot (k-i)} = p_k {k \choose i} \ . \label{eqn:04}
\end{align}

We can now bound $p_k$ using $\mathbf{C2}$ and Inequality~\ref{eqn:04} as follows:
$$p_k = 1 - \sum_{i=1}^{k-1} p_i \ge 1-  \sum_{i=1}^{k-1} p_k {k \choose i} = 1 - p_k (2^k - 2) \ ,$$
using the identity $\sum_{i=0}^{k} {k \choose i} = 2^k$ and ${k \choose 0} = {k \choose k} = 1$. We thus obtain $p_k \ge \frac{1}{2^k - 1}$ as claimed.

Next, we bound the expectation
\begin{align}
\Exp_{\pi \mid C} |OPT_B \cap P| = \sum_{i=1}^{k} i \cdot p_i \ . \label{eqn:05}
\end{align}
Let $q_1, \dots, q_k$ denote values for $p_1, \dots, p_k$ that minimize Equation~\ref{eqn:05} and adhere to the constraints $\mathbf{C1}$ and $\mathbf{C2}$. We claim that $q_{i+1} = q_i \cdot \frac{k-i}{i+1}$ must hold, for all $1 \le i \le k-1$, i.e., the constraints $\mathbf{C1}$ are tight and become equalities for $q$. To see this, suppose that this is not true, i.e., there is at least one index such that $\mathbf{C1}$ is not tight. Let $1 \le j \le k-1$ be the largest such index, i.e., we have
$$q_{j+1} > q_j \cdot \frac{k-j}{j+1} \ . $$ Then, we define $\epsilon > 0$ such that 
\begin{align}
q_{j+1} =  q_j \cdot \frac{k-j}{j+1} + \epsilon  \ . \label{eqn:392} 
\end{align}
We now argue that $(q_i)_{1 \le i \le k-1}$ does not minimize the expected value stated in Equation~\ref{eqn:05}. To this end, we define an alternative solution $(r_i)_{1 \le i \le k-1}$ such that: 
\begin{align*}
 r_i & = q_i, \mbox{ for all } i \neq \{j,j+1\},  \\
 r_j & = q_j + \epsilon \cdot \frac{j+1}{k+1} \ , \text{ and} \\
 r_{j+1} & = %
 q_{j+1} - \epsilon \cdot  \frac{j+1}{k+1} \ . %
\end{align*}
We observe that $r_j + r_{j+1} = q_j + q_{j+1}$, and since for every $i \notin \{j,j+1\}$, we have $r_i = q_i$, we have that $\sum_{i \ge 1}^{k} r_i = \sum_{i \ge 1}^{k} q_i = 1$, i.e., $(r_i)_{1 \le i \le k}$  fulfills constraint $\mathbf{C2}$.

Next, regarding $\mathbf{C1}$, since we only modified the indices $j$ and $j+1$, we need to check $\mathbf{C1}$ for any index $i$ such that $\{i, i+1\} \cap \{j, j+1\} \neq \{ \}$, which are the indices $i \in \{j-1, j, j+1\}$. We thus have to verify the following constraints:
\begin{align}
 q_{j+2} = r_{j+2} & \ge r_{j+1} \cdot \frac{k-(j+1)}{(j+1)+1} \ ,  \label{eqn:C1-1} \\
 r_{j+1} & \ge r_{j} \cdot \frac{k-j}{j+1} \ , \text{ and} \label{eqn:C1-2} \\
 r_{j} & \ge r_{j-1} \cdot \frac{k-(j-1))}{(j-1)+1}  = q_{j-1} \cdot \frac{k-(j-1)}{(j-1)+1}  \ . \label{eqn:C1-3}
\end{align}
Regarding Inequality~\ref{eqn:C1-1}, observe that $r_{j+1} < q_{j+1}$, which implies that this inequality holds, and
regarding Inequality~\ref{eqn:C1-3}, observe that that $r_j > q_j$, which implies that this inequality holds. We now verify that equality holds in Inequality~\ref{eqn:C1-2}:
\begin{align*}
 r_{j+1} & = q_{j+1} - \epsilon \cdot  \frac{j+1}{k+1} & \text{Definition of $r_{j+1}$}\\ 
 & =  \left( q_j \cdot \frac{k-j}{j+1} + \epsilon \right) - \epsilon \cdot  \frac{j+1}{k+1} & \text{Equality~\ref{eqn:392}} \\
 & = \left(r_j -  \epsilon \cdot \frac{j+1}{k+1} \right) \cdot \frac{k-j}{j+1} + \epsilon \left( 1 - \frac{j+1}{k+1} \right) & \text{Definition of $r_j$} \\
 & = r_j \cdot \frac{k-j}{j+1} + \epsilon \left( 1 - \frac{j+1}{k+1} - \frac{k-j}{k+1} \right) \\
 & = r_j \cdot \frac{k-j}{j+1} \ .
\end{align*}
Last, we observe that $(r_i)_{1 \le i \le k}$ yields a smaller expected value as in Inequality~\ref{eqn:05} as $(q_i)_{1 \le i \le k}$ since:
\begin{align*}
(j+1) \cdot r_{j+1}  + j \cdot r_j & = (j+1) \cdot r_{j} \cdot \frac{k-j}{j+1} + j \cdot r_j \\
& = k \cdot r_j \\
& = k \cdot \left(q_j + \epsilon \cdot \frac{j+1}{k+1} \right) \\
& < k q_j + \epsilon (j+1) \\
& = (j+1) \left(q_j \frac{k-j}{j+1} + \epsilon \right) + j \cdot q_j \\
& = (j+1) \cdot q_{j+1}  + j \cdot q_j \ .
\end{align*}
We thus conclude that, for the minimizer $(q_i)_{1 \le i \le k}$, constraint $\textbf{C1}$ is tight, for all $1 \le i \le k-1$. This, in turn, implies that $q_i =  {k \choose i} q_k$, for every $1 \le i \le k-1$. 

We thus obtain:
\begin{align*}
\Exp_{\pi \mid C} |OPT_B \cap C| & = \sum_{i=1}^{k} i \cdot p_i \ge \sum_{i=1}^{k} i \cdot q_k \cdot {k \choose i} = q_k \sum_{i=1}^{k} i \cdot {k \choose i} = q_k \left( \left(\sum_{i=1}^{k-1} i \cdot {k \choose i} \right) + k \right)  \\
& = k q_k + q_k  \frac{k}{2} \sum_{i=1}^{k-1} {k \choose i}  = k q_k + q_k  \frac{k}{2} ( 2^{k}-2 ) = k q_k \cdot 2^{k-1} \ge k \frac{2^{k-1}}{2^k - 1} \ .
\end{align*}
In the previous calculation, we used the identity $\sum_{i=1}^{k-1} i \cdot {k \choose i} = \frac{k}{2} \sum_{i=1}^{k-1} {k \choose i}$, which follows from the fact that, for any $i$, $i \cdot {k \choose i} + (k-i) \cdot {k \choose k-i} = \frac{k}{2} \cdot {k \choose i} + \frac{k}{2} \cdot {k \choose k-i}$. This completes the proof.
\end{proof}

\subsection{Bounding $OUT$ and $OPT$}\label{sec:bounds-out-opt}

Using Lemma~\ref{lem:main}, we can now bound the size of $OUT$ as follows:

\begin{lemma} \label{lem:bound-out}
The output matching $OUT$ produced by Bob is bounded in size from below as follows:
 \begin{align*}
\Exp |OUT| & \ge \sum_{i \ge 1} \left( \frac{1}{2^i - 1}  + i-1 \right) n_i  \ . 
 \end{align*}
\end{lemma}
\begin{proof}
 Let $C \in supp(\mathcal{C})$ be any possible outcome for $\mathcal{C}$.
 
 We partition the components $C$ into odd-length paths $(P_i)_{i \ge 1}$ such that each path in $P_i$ consists of $2i - 1$ edges, and even-length paths or even-length cycles $(E_i)_{i \ge 1}$ such that each path/cycle in $E_i$ consists of $2i$ edges. For a component $D \in C$, we write $OUT_{D}$ to denote Bob's output on $D$, and by $OPT_{D}$ the optimal edges of $D$.
 
 Consider an isolated edge $\{e \} = D \in P_1$. We observe that $e$ must be an optimal edge since if it was not an optimal edge then there would be an optimal edge incident on $e$, which implies that the component containing $e$ would be of size at least $2$. Next, suppose that $e$ is in $OPT_A$. Then, Alice has included $e$ in $M$ since otherwise the matching $M$ would not have been a maximum matching, which implies that Bob knows $e$ as it was sent to Bob as part of the message $M$. Otherwise, if $e$ was in $OPT_B$ then Bob knows $e$ as well as it is part of his input. In either case, Bob can output a matching of optimal size $1$ on each path in $P_1$, i.e., $|OUT_{D}| = |OPT_{D}|$.
 
 Next, consider any even-length component $D \in E_i$ of length $2i$, for some $i$. Since such a component contains $i$ non-optimal edges that are contained in $M$, Bob can always output a matching of (optimal) size $i$ on such a component, and we thus also obtain $|OUT_{D}| = |OPT_{D}|$.
 
 The only loss incurred takes place on odd-length paths $P_i$, for some $i \ge 2$.  
  Then, as proved in Lemma~\ref{lem:main}, conditioned on $C$ existing, Bob holds all optimal edges in this path with probability $\frac{1}{2^i-1}$. Hence, in expectation, Bob can output a solution of size:
 \begin{align*}
 \Exp_{\pi \mid C} |OUT_{D}| \ge \frac{1}{2^i-1} \cdot i + (1 - \frac{1}{2^i-1}) \cdot (i-1) = \frac{1}{2^i-1}+i - 1 \ .
 \end{align*}

 We now piece all of the above together in the following:
 \begin{align*}
  \Exp_{\pi} |OUT| & = \sum_{C \in supp(\mathcal{C})} \Pr \left[\mathcal{C} = C \right] \cdot \Exp_{\pi | C} |OUT|   & \text{Law of total expectation} \\
  & =  \sum_{C \in supp(\mathcal{C})} \Pr \left[\mathcal{C}=C \right] \cdot \sum_{D \in C} \Exp_{\pi | C} |OUT_{D}| \ . & \text{Linearity of expectation}
 \end{align*}
 It remains to bound $\sum_{D \in C} \Exp_{\pi | C}  |OUT_D|$, for any fixed $C \in supp(\mathcal{C})$. We obtain:
 \begin{align*}
\sum_{D \in C} \Exp_{\pi | C}  |OUT_D| & = \underbrace{\left( \sum_{i \ge 1}  \sum_{D \in E_i} i \right) +  \left( \sum_{D \in P_1}1 \right)}_{= n_1(C)} + \left(\sum_{i \ge 2} \sum_{D \in P_i} \frac{1}{2^i-1}+i - 1 \right) \\
& = n_1(C) + \sum_{i \ge 2} \left( \frac{1}{2^i-1}+i - 1 \right) n_i(C) \ .
 \end{align*}
 Thus, overall, we obtain:
 \begin{align*}
  \Exp_{\pi} |OUT| & = \sum_{C \in supp(\mathcal{C})} \Pr \left[\mathcal{C} = C \right] \cdot \left( n_1(C) + \sum_{i \ge 2} \left( \frac{1}{2^i-1}+i - 1 \right) n_i(C)  \right) \\
  & = n_1 + \sum_{i \ge 2} \left( \frac{1}{2^i-1}+i - 1 \right) n_i \\
  & = \sum_{i \ge 1} \left( \frac{1}{2^i-1}+i - 1 \right) n_i \ .
 \end{align*}

\end{proof}

\begin{lemma}\label{lem:bound-opt}
The following bounds on the size of an optimal matching hold:
\begin{align}
 \sum_{i \ge 2} i \cdot \frac{2^{i}}{2^i - 1} \cdot  n_i  & \le |OPT| \ , \text{ and} \label{eqn:291} \\ 
 \sum_{i \ge 1} i \cdot n_i & = |OPT| \label{eqn:292} \ .
\end{align}
\end{lemma}

\begin{proof}
Let $C \in supp(\mathcal{C})$ be any possible outcome. 

As in the proof of the previous lemma, we consider the partitioning of the components $C$ into odd-length paths $(P_i)_{i \ge 1}$ such that each path in $P_i$ consists of $2i - 1$ edges and even-length paths or even-length cycles $(E_i)_{i \ge 1}$ such that each path/cycle in $E_i$ consists of $2i$ edges.

We argue the lower bound on $|OPT|$ stated in Inequality~\ref{eqn:291} first.

Since every edge in $OPT$ is also contained in $OPT_B$ with probability $\frac{1}{2}$, by linearity of expectation, we obtain
 \begin{align*}
  \Exp |OPT_B| & = |OPT| / 2 \ .
 \end{align*}
Next, as proved in Lemma~\ref{lem:main}, for each odd-length path, in expectation, Bob holds more than half of the $OPT$ edges. We use this to give a lower bound on the expected size of $OPT_B$. 
For the set of components $C$, we denote by $P_i(C)$ the odd-length paths of length $2i-1$. We obtain:
\begin{align*}
 |OPT| / 2 = \Exp |OPT_B| & \ge \sum_{C \in supp(\mathcal{C})} \Pr \left[\mathcal{C} = C \right] \cdot \left( \sum_{i \ge 2} \sum_{D \in P_i(C)} \Exp_{\pi \mid C} |OPT_B \cap D| \right) \\
 & \ge \sum_{C \in supp(\mathcal{C})} \Pr \left[\mathcal{C} = C \right] \cdot \left( \sum_{i \ge 2} \sum_{D \in P_i(C)} i \cdot \frac{2^{i-1}}{2^i -1} \right) & \text{Lemma~\ref{lem:main}} \\
 & \ge \sum_{C \in supp(\mathcal{C})} \Pr \left[\mathcal{C} = C \right] \cdot \left( \sum_{i \ge 2} n_i(C) \cdot  i \cdot \frac{2^{i-1}}{2^i -1} \right) \\
 & = \sum_{i \ge 2} n_i  \cdot i \cdot \frac{2^{i-1}}{2^i -1} \ ,  & \text{Equality~\ref{eqn:966}}
\end{align*}
which implies the result.

Next, we argue Equality~\ref{eqn:292}. We observe that each component in $E_i$ contains $i$ edges from $OPT$, and each component in $P_i$ also contains $i$ edges from $OPT$. Thus, overall, we obtain that for any set of components in $C \in supp(\mathcal{C})$:
\begin{align*}
 |OPT| = n_1(C) + \sum_{i \ge 2} i \cdot n_i(C) = \sum_{i \ge 1} i \cdot n_i(C) \ .
\end{align*}
Recall that, by definition of $n_1(C)$, $n_1(C)$ counts all the optimal edges in even-length components in $C$ as well as in paths of length $1$ in $C$. Equality~\ref{eqn:292} follows by taking the expected value on the previous equality and observing that $\Exp_\pi |OPT| = |OPT|$. This completes the proof.
\end{proof}

\subsection{Our Main Result}\label{sec:main-thm}
We are now ready to prove our main result stated in the following theorem:

\begin{figure}
 \begin{center} \fbox{\begin{minipage}{0.8\textwidth}
 \begin{align}
   \text{minimize } \sum_{i \ge 1} & \left(\frac{1}{2^i-1} + i -  1 \right)  m_i  \nonumber  \\
  \text{subject to} \hspace{0.1cm}\sum_{i \ge 2} i \cdot \frac{2^{i}}{2^i - 1} \cdot  m_i & \le 1  \label{eqn:lp-1} \\
   \sum_{i\ge 1}i \cdot  m_i & = 1 \label{eqn:lp-2} \\
   \text{For all $i$:} \quad m_i & \ge 0  \ . 
 \end{align} \end{minipage}} \end{center}
 \caption{LP whose objective value constitutes a lower bound on the approximation factor of $\Pi_{\text{LFMM}}$. \label{fig:lp}}
\end{figure}

\begin{theorem}
 The protocol $\Pi_{\text{LFMM}}$ has an expected approximation factor of at least $3/4$.
\end{theorem}
\begin{proof}
For every $i \ge 1$, let $m_i = n_i / |OPT|$ be the normalized version of $n_i$. Then, by Lemma~\ref{lem:bound-out}, the expected approximation factor of $\Pi_{\text{LFMM}}$ can be bounded as follows:
\begin{align}
 \Exp \frac{|OUT|}{|OPT|} & = \frac{\Exp |OUT|}{|OPT|} \geq  \frac{\sum_{i \ge 1} \left( \frac{1}{2^i -1} + i - 1 \right) n_i}{|OPT|} = \sum_{i \ge 1}  \left(\frac{1}{2^i -1} + i - 1 \right) m_i \ .
\end{align}
Our aim is to identify the worst-case assignment to the values $(m_i)_{i \ge 1}$ that minimize the expected approximation factor. From Lemma~\ref{lem:bound-opt}, we obtain the constraints (by dividing the inequality by $|OPT|$):
\begin{align}
 \sum_{i \ge 2} i\cdot\frac{2^i}{2^i - 1} \cdot m_i & \le 1  \ , \text{ and} \\ 
 \sum_{i \ge 1} i \cdot m_i & = 1 \ .
\end{align}
We thus see that the approximation factor of our protocol is bounded from below by the objective value of the linear program illustrated in Figure~\ref{fig:lp}. 

We will show in Lemma~\ref{lem:lp} that there exists an assignment that minimizes the objective value such that $m_i = 0$, for all $i \ge 3$ holds. Consider thus such a solution. Then, we obtain the following simplified LP: Minimize $m_1 + \frac{4}{3} m_2$ subject to $m_1, m_2$ being non-negative numbers such that $\frac{8}{3} m_2 \le 1$ and $m_1 + 2 m_2 = 1$. These constraints yield that $m_2 = \frac{3}{8}$ and $m_1 = \frac{1}{4}$ minimize the objective value, which amounts to $m_1 + \frac{4}{3} m_2 = \frac{1}{4} + \frac{4}{3} \cdot \frac{3}{8} = \frac{3}{4}$. This completes the proof.
\end{proof}

\begin{lemma} \label{lem:lp}
 Consider the linear program of Figure~\ref{fig:lp}. Then, there exists an assignment $m'_1, m'_2, m'_3, \dots$ that minimizes the objective value such that $m'_i = 0$, for every $i \ge 3$. 
\end{lemma}
\begin{proof}
 Let $m_1, m_2, \dots$ denote a solution that minimized the objective value, and suppose that there exists an index $j \ge 3$ such that $m_j > 0$. Then, we define an alternative solution $m'_1, m'_2, \dots$ as follows: 
 \begin{align*}
  \text{For every } i \notin \{1, 2, j\}: m_i' & = m_i \ , \\
  m'_j & = 0 \ , \\
  m'_2 & = m_2 + \frac{3}{8} m_j \cdot j \cdot \frac{2^j}{2^j - 1} \ , \text{ and } \\ 
  m_1' &= m_1 +  j \cdot m_j \cdot(1 - \frac{3}{4} \cdot  \frac{2^j}{2^j -1}) \ . 
 \end{align*}
  We verify that $m'$ is a valid solution of the linear program of Figure~\ref{fig:lp} with an objective value that is at most the one given by $m$. To this end, consider first Inequality~\ref{eqn:lp-1}. We have:
  \begin{align*}
   \sum_{i \ge 2} i \cdot \frac{2^{i}}{2^i - 1} \cdot  m'_i & = 2 \cdot \frac{4}{3} \cdot \left(m_2 + \frac{3}{8} m_j \cdot j \cdot \frac{2^j}{2^j - 1}\right) + \sum_{i \ge 3, i \neq j} i \cdot \frac{2^i}{2^i - 1} \cdot m_i \\
   & = \sum_{i \ge 2} i \cdot \frac{2^i}{2^i-1} m_i \le 1 \ .
  \end{align*} 
Next, we verify Equality~\ref{eqn:lp-2}. We have:
  \begin{align*}
   m_1' + \sum_{k \ge 2} k \cdot m_k' & = \left( m_1 +  j \cdot m_j \cdot(1 - \frac{3}{4} \cdot  \frac{2^j}{2^j -1})\right) + 2 \cdot \left(m_2 + \frac{3}{8} m_j \cdot j \cdot \frac{2^j}{2^j - 1}\right) + \sum_{i \ge 3, i \neq j} i \cdot m_i \\
   & =  m_1 + 2 \cdot m_2 + j \cdot m_j + \left(\sum_{i \ge 3, i \neq j} m_i \right) = m_1 + \sum_{k \ge 2} k \cdot m_k = 1 \ .
  \end{align*}
Last, regarding the objective value, we have: 
\begin{align*}
 \sum_{i \ge 1}  & \left(\frac{1}{2^i-1} + i -  1 \right)  m'_i = \\
 & \left(m_1 +  j \cdot m_j \cdot(1 - \frac{3}{4} \cdot  \frac{2^j}{2^j -1}) \right) + \frac{4}{3} \cdot \left( m_2 + \frac{3}{8} m_j \cdot j \cdot \frac{2^j}{2^j - 1} \right) + \sum_{i\ge 3, i \neq j} \left(\frac{1}{2^i-1} + i -  1 \right)  m_i \\
 & = \left( \sum_{k\ge 1, k \neq j} \left(\frac{1}{2^k-1} + k -  1 \right)  m_k \right) +  m_j \cdot \underbrace{j \cdot \left( 1 - \frac{2^{j-2}}{2^j -1}  \right)}_{:=X} \ . %
\end{align*}
It remains to show that $X \le \frac{1}{2^j-1}+j-1$. We compute:
\begin{align*}
 X = j \cdot \left( 1 - \frac{2^{j-2}}{2^j -1}  \right) = j \cdot \frac{3 \cdot 2^{j-2} - 1}{2^j - 1} & \stackrel{!}{\le} \frac{1}{2^j-1}+j-1 \quad \Leftrightarrow \\
 j \cdot 3 \cdot 2^{j-2} - j & \le 1 + j \cdot 2^j - j - 2^j + 1 \\
 j \cdot 3 \cdot 2^{j-2} & \le 2 + j 2^j - 2^j \\
 2^j & \le 2 + j \cdot 2^{j-2} \ ,
\end{align*}
which holds for every $j \ge 3$.
\end{proof}

\subsection{Sending a Lexicographically-first Maximal Matching}
Our analysis can easily be adapted to the alternative protocol $\Pi'$ where Alice sends a lexicographically-first {\em maximal matching} to Bob, which achieves an expected approximation factor of $5/8$, and this is also tight. 

The approximation factor of $5/8$ can be seen as follows. Via a similar analysis as given in Lemma~\ref{lem:main}, it can be seen that, for each augmenting-path of length $3$ in the graph $H$ spanned by the lexicographically-first maximal matching $M$ on Alice's edges $E_A$ and the optimal edges $OPT$, the probability that Bob holds both optimal edges is at least $1/4$. This bound is worse than the $1/3$ bound that we established for the $\Pi_{\text{LFMM}}$ protocol, which is due to the fact that in $\Pi_{\text{LFMM}}$, it is not possible that Bob holds no optimal edges while this can be the case in $\Pi'$. 

To complete the analysis of $\Pi'$, Bob thus achieves a matching of expected size $\frac{3}{4} \cdot 1 + \frac{1}{4} \cdot 2 = \frac{5}{4}$ on each augmenting path of length $3$ in $H$, which implies that the expected approximation factor on such a path is $\frac{5}{2\cdot 4} = \frac{5}{8}$. Last, we observe that, for any other component (either even-length paths or cycles, single edges, or odd-length paths of length at least $5$), Bob trivially achieves a better than $5/8$-approximation, which establishes the expected approximation factor of $\frac{5}{8}$ of $\Pi'$.

Last, we remark that the graph consisting of a single path of length $3$ where the middle edge is given to Alice yields an expected approximation factor of $5/8$ for $\Pi'$.

\section{Conclusion} \label{sec:conclusion}
In this paper, we showed that the $\Pi_{\text{LFMM}}$ protocol achieves a $3/4$-approximation in expectation in the semi-robust setting and that our analysis is tight, in that there are graphs on which an expectation approximation factor of $3/4$ is achieved. We also showed that if Alice sends a lexicographically-first maximal matching then an approximation factor of $5/8$ is achieved, which is also tight.

We conclude with two avenues for further research:
\begin{enumerate}
 \item What is the expected approximation factor of $\Pi_{\text{LFMM}}$ in the fully robust setting? Are there variants of the protocol, e.g., sending a uniform random maximum matching, that achieve the state-of-the-art bound of $5/6$ or potentially even beat this bound?
 \item Are there either variants of the $\Pi_{\text{LFMM}}$ protocol or entirely different approaches that yield a better than $3/4$-approximation factor in the semi-robust setting?
\end{enumerate}

\bibliographystyle{plain}
\bibliography{bibliography}

\end{document}